\documentclass{article}

\usepackage[ansinew]{inputenc}
\usepackage{amssymb,amsmath,amsthm}
\usepackage{graphicx}
\usepackage{url}
\usepackage[hmargin=1.64in]{geometry}
\usepackage{tabularx}
\usepackage{tikz}
\usetikzlibrary{calc}
\usetikzlibrary{decorations.markings}
\usepackage{subfig}

\usepackage{xspace}
\usepackage{vmargin}
\setmarginsrb{1.1in}{1.1in}{1.1in}{1.1in}{0cm}{0cm}{8mm}{8mm}

\newtheorem{theorem}{Theorem}

\newtheorem{observation}[theorem]{Observation}

\newtheorem{lemma}[theorem]{Lemma}

\newtheorem{boldclaim}{Claim}
\newtheorem{corol}{Corollary}

\newcommand{\fvstub}{1.6740}
\newcommand{\FVS}{FVS\xspace}
\newcommand{\FVSs}{FVSs\xspace}

\title{Feedback Vertex Sets in Tournaments\thanks{Part of this research has been supported by the Netherlands Organisation for Scientific Research (NWO), grant 639.033.403. A preliminary version of this article appeared in the proceedings of ESA 2010 \cite{GaspersM10}.}}

\author{
 Serge Gaspers\thanks{%
  Inst.\ of Information Systems,
  Vienna University of Technology,
  Vienna, Austria.
  E-mail: \texttt{gaspers@kr.tuwien.ac.at}}
 \and Matthias Mnich\thanks{%
  Technische Universiteit Eindhoven, Eindhoven, The Netherlands.
  E-mail: \texttt{m.mnich@tue.nl}}
}

\date{}

\begin{document}

\maketitle

\begin{abstract}
We study combinatorial and algorithmic questions around minimal feedback vertex sets in tournament graphs.

On the combinatorial side, we derive upper and lower bounds on the maximum number of minimal feedback vertex sets in an $n$-vertex tournament.
We prove that every tournament on $n$ vertices has at most $\fvstub^n$ minimal feedback vertex sets, and that there is an infinite family of tournaments, all having at least $1.5448^n$ minimal feedback vertex sets.
This improves and extends the bounds of Moon (1971).

On the algorithmic side, we design the first polynomial space algorithm that enumerates the minimal feedback vertex sets of a tournament with polynomial delay.
The combination of our results yields the fastest known algorithm for finding a minimum size feedback vertex set in a tournament.
\end{abstract}


\section{Introduction}
\label{sec:introduction}
A tournament $T=(V,A)$ is a directed graph with exactly one arc between every pair of vertices.
A feedback vertex set (\FVS) of $T$ is a subset of its vertices whose deletion makes $T$ acyclic.
A minimal \FVS of $T$ is a \FVS of $T$ that is minimal with respect to vertex-inclusion.
The complement of a minimal \FVS~$F$ induces a maximal acyclic subtournament whose unique vertex with no in-neighbor is a ``Banks winner''~\cite{Banks1985}: identifying the vertices of $T$ with candidates in a voting scheme and arcs indicating preference of one candidate over another, the \emph{Banks winner} of $T[V\setminus F]$ is the candidate collectively preferred to every other candidate in~$V\setminus F$.
Banks winners play an important role in social choice theory.
Minimal \FVSs are also related to so-called ``stable sets'' in tournaments \cite{Brandt2011}, a notion which is inspired by stable sets in game theory \cite{vonNeumannM44}. They could play a major role in proving a conjecture of Brandt~\cite{Brandt2011}.

\paragraph{Extremal Combinatorics.}
We denote the number of minimal \FVSs in a tournament $T$ by~$f(T)$, and the maximum $f(T)$ over all $n$-vertex tournaments by~$M(n)$.
The letter ``M'' was chosen in honor of Moon who in 1971 proved~\cite{Moon1971} that
\begin{equation*}
  1.4757^n \leq  M(n) \leq 1.7170^n 
\end{equation*}
for large $n$.
Our combinatorial main result are the stronger bounds
\begin{equation*}
  1.5448^n \leq  M(n) \leq \fvstub^n \enspace .
\end{equation*}
To prove our new lower bound on  $M(n)$, we construct an infinite family of tournaments all having $21^{n/7} > 1.5448^n$ minimal \FVSs.
To prove our new upper bound on  $M(n)$, we bound the maximum of a convex function bounding $M(n)$ from above, and otherwise rely on case distinctions and recurrence relations.

For general directed graphs, no non-trivial upper bound on the number of minimal \FVSs is known. For undirected graphs, Fomin et al.~\cite{FominEtAl2008} show that any undirected graph on~$n$ vertices contains at most $1.8638^n$ minimal \FVSs, and that infinitely many graphs have $105^{n/10} > 1.5926^n$ minimal \FVSs.
Lower bounds of roughly $\log n$ on the size of a maximum-size acyclic subtournament have been obtained by Reid and Parker \cite{ReidParker1970} and Neumann-Lara \cite{NeumannLara94}.
Other bounds on minimal or maximal sets with respect to vertex-inclusion have been obtained for dominating sets~\cite{FominGPS08}, bicliques~\cite{GaspersKL08}, separators~\cite{FominV08}, potential maximal cliques~\cite{FominV10}, bipartite graphs~\cite{ByskovMS05}, $r$-regular subgraphs~\cite{GuptaRS06}, and, of course, independent sets~\cite{MillerM60,MoonM65}.
The increased interest in exponential time algorithms over the last few years has given new importance to such bounds, as the enumeration of the corresponding objects may be used in exponential time algorithms to solve various problems; see, for example \cite{BjorklundHK09,Byskov04,Eppstein03,FominV10,Lawler76,RamanSS07}.

\paragraph{Enumeration.}
An algorithm by Schwikowski and Speckenmeyer~\cite{SchwikowskiEtAl2002} lists the minimal \FVSs of a directed graph $G$ with polynomial delay, by traversing a hypergraph whose vertices are bijectively mapped to minimal \FVSs of $G$.
Unfortunately the Schwikowski-Speckenmeyer-algorithm may use exponential space, and it is not known whether the minimal FVS problem allows a polynomial delay enumeration algorithm with polynomially bounded space complexity in directed graphs.
Our algorithmic main result provides such an enumeration algorithm for the family of \emph{tournaments}.
Our algorithm is inspired from that by Tsukiyama et al. for the (conceptually simpler) enumeration of maximal independent sets~\cite{TsukiyamaEtAl1977}.
It is based on iterative compression, a technique for parameterized \cite{ReedSV04} and exact algorithms \cite{FominGKLS08}.
We thereby positively answer Fomin et al.'s~\cite{FominGKLS08} question regarding if the technique could be applied to other algorithmic areas.

\paragraph{Exact Algorithms.}
In the third \cite{Woeginger2008} in a series \cite{Woeginger03,Woeginger04,Woeginger2008} of very influential surveys on exact exponential time algorithms, Woeginger observes that Moon's upper bound on $M(n)$ provides an upper bound on the overall running time of the enumeration algorithm of Schwikowski and Speckenmeyer.
He explicitly asks for a faster algorithm for finding a feedback vertex set of minimum size in a tournament.
Our new bound yields a time complexity of $O(1.6740^n)$. Unlike upper bound proofs on other \cite{ByskovMS05,FominEtAl2008,FominGPS08,FominV08,FominV10,GaspersKL08,GuptaRS06,MillerM60,MoonM65} minimal or maximal sets with respect to vertex inclusion, for minimal \FVSs in tournaments no known (non trivial) proof readily translates into a polynomial-space branching algorithm.
Due to its space complexity, which differs from its time complexity by only a polynomial factor, the Schwikowski-Speckenmeyer-algorithm has only limited practicability \cite{Woeginger2008}.
With our new enumeration algorithm, we achieve however a polynomial-space $O(1.6740^n)$-time algorithm to find a minimum sized feedback vertex set in tournaments, and even to enumerate all minimal ones.
Dom et al. \cite{DomEtAl2006} independently answered Woeginger's question by constructing an iterative--compression algorithm solving only the optimization version of the problem.
However, the running time of their algorithm grows at least with $1.708^n$ and hence their result is inherently weaker than ours.


\paragraph{Organization of the paper.}
Preliminaries are provided in Section~\ref{sec:preliminaries}.
In Section~\ref{sec:minimum_number}, we answer how many distinct minimal \FVSs a (strong) tournament on $n$ vertices has \emph{at least}.
Section~\ref{sec:lowerbound} proves the lower bound on $M(n)$, and Section \ref{sec:upperbound} gives the upper bound.
We conclude with the polynomial-space po\-ly\-no\-mial-delay enumeration algorithm in Section~\ref{sec:polydelaypolyspace}.
The main result of the paper is formulated in Corollary~\ref{thm:minfvspolyspace}.

\section{Preliminaries}
\label{sec:preliminaries}
Let $T = (V,A)$ be a tournament.
For a vertex subset $V'\subseteq V$, the tournament $T[V']$ induced by $V'$ is called a \emph{subtournament} of~$T$.
For each vertex $v\in V$, its \emph{in-neighborhood} and \emph{out-neighborhood} are defined as $N^-(v)=\{u\in V~|~(u,v)\in A\}$ and $N^+(v)=\{u\in V~|~(v,u)\in A\}$, respectively.
If there is an arc $(u,v)\in A$ then we say that $u$ \emph{dominates} $v$ and write $u \rightarrow v$.
A tournament is \emph{strongly connected}, or simply \emph{strong}, if there exists a directed path between any two distinct vertices.
A non-strong tournament $T$ has a unique factorization $T = S_1 + \hdots + S_r$ (or Zykov sum) into strong subtournaments $S_1,\hdots,S_r$, where every vertex $u\in V(S_k)$ dominates all vertices $v\in V(S_\ell)$, for $1\leq k < \ell\leq r$.
For $n\in\mathbb N$  let $\mathcal T_n$ denote the set of tournaments with $n$ vertices and let $\mathcal T^*_n$ denote the set of strong tournaments on $n$ vertices.

The \emph{score} of a vertex $v\in V$ is the size of its out-neighborhood, and denoted by $s_v(T)$ or $s_v$ for short.
Consider a labeling $1,\hdots,n$ of the vertices of $T$ such that their scores are non-decreasing, and associate with $T$ the \emph{score sequence} $s(T)=(s_1,\hdots,s_n)$.
If $T$ is strong then $s(T)$ satisfies \emph{Landau's inequalities}~\cite{HararyMoser1966,Landau1953}:
\begin{align}
\sum_{v=1}^k s_v &\geq \binom{k}{2}+1~\mbox{ for all }~k=1,\hdots,n-1, \text{ and}
\label{eqn:sbound2}\displaybreak[0]\\
\sum_{v=1}^n s_v &= \binom{n}{2}
\label{eqn:sbound3} \enspace .
\end{align}
For every non-decreasing sequence $s$ of positive integers satisfying conditions \eqref{eqn:sbound2}--\eqref{eqn:sbound3}, there exists a tournament whose score sequence is $s$~\cite{Landau1953}.

Let $L$ be a set of non-zero elements from the ring $\mathbb Z_n$ of integers modulo $n$ such that for all $i\in \mathbb Z_n$ exactly one of $+i$ and $-i$ belongs to $L$.
The tournament $T_L = (V_L,A_L)$ with $V_L = \{1,\hdots,n\}$ and $A_L = \{(i,j)\in V_L\times V_L ~|~(j-i)\bmod n~\in L\}$ is \emph{the circular n-tournament induced by} $L$.
A \emph{triangle} is a tournament of order~$3$. The cyclic triangle is denoted $C_3$.

A \emph{feedback vertex set (FVS)} of a tournament $T=(V,A)$ is a subset $F$ of vertices such that $T[V\setminus F]$ has no directed cycle.
It is \emph{minimal} if it does not contain a \FVS of $T$ as a proper subset.
Let $\mathcal F(T)$ be the collection of minimal \FVSs of $T$; its cardinality is denoted by $f(T)$.
A \emph{minimum \FVS} is a \FVS with the least possible number of vertices.

Acyclic tournaments are sometimes called \emph{transitive}; the (up to isomorphism unique) transitive tournament on $n$ vertices is denoted $TT_n$.
Let $\tau$ be the unique topological order of the vertices of $TT_n$ such that $\tau(u) < \tau(v)$ if and only if $u$ dominates $v$.
For such an order $\tau$ and integer $i\in\{1,\hdots,n\}$ the subsequence of the first $i$ values of $\tau$ is denoted $\tau_i(V(TT_n))=(\tau^{-1}(1),\hdots,\tau^{-1}(i))$; call $\tau_1(V(TT_n))$ the \emph{source} of $TT_n$.
For a minimal \FVS $F$ of a tournament $T$ the subtournament $T[V\setminus F]$ is a \emph{maximal transitive subtournament} of $T$ and $V \setminus F$ is a \emph{maximal transitive vertex set}.

\section{Minimum Number of Minimal FVSs}
\label{sec:minimum_number}
In this section we analyze the minimum number of minimal \FVSs in tournaments.

Let the function $m:\mathbb N\rightarrow\mathbb N, n\mapsto \min_{T\in\mathcal T_n}f(T)$ count the minimum number of minimal \FVSs over all tournaments of order $n$.
Since a minimal \FVS always exists, $m(n)\geq 1$ for all positive integers $n$.
This bound is attained, for every $n \in \mathbb{N}$, by the transitive tournaments $TT_n$.

\begin{observation}[\cite{Moon1971}]
\label{thm:fvsstrongdecomposition}
  If $T = S_1 + \hdots + S_r$ is the factorization of a tour\-na\-ment $T$ into strong subtournaments $S_1,\hdots,S_r$, then $f(T)=f(S_1)\cdot\hdots\cdot f(S_r)$.
\end{observation}
Hence from now on we consider only strong tournaments (on at least $3$ vertices) and define
$m^*:\mathbb N\setminus\{1,2\}\rightarrow\mathbb N, n\mapsto \min_{T\in\mathcal T^*_n}f(T)$.
\begin{lemma}
\label{thm:constmstar}
The function $m^*$ is constant: $m^*(n) = 3$ for all $n\geq 3$.
\end{lemma}
\begin{proof}
Let $T\in\mathcal T^*_n$ be a strong tournament. We show that $f(T)\geq 3$.
As $T$ is strong, it contains some cycle and thus some cyclic triangle \cite{Moon66}, with vertices $v_1,v_2,v_3$.
For $i=1,2,3$, define the vertex sets $W_i=\{v_i,v_{(i+1)\mod 3}\}$.
Every set $W_i$ can be extended to a maximal transitive vertex set $W_i'$ of $T$.
Note that for $i=1,2,3$ and $j\in\{1,2,3\}\setminus\{i\}$, we have $v_{(i+2)\mod 3}\in W_j'\setminus W_i'$.
Hence, there are three maximal transitive subtournaments of $T$ whose complements form three minimal \FVSs of $T$.
Consequently, $m^*(n)\geq 3$ for all $n\geq 3$.

To complete the proof, construct a family $\{U_n\in\mathcal T^*_n~|~n\geq 3\}$ of strong tournaments with exactly three minimal \FVSs.
Set $U_3$ equal to the cyclic triangle.
For $n\geq 4$, build the tournament $U_n$ as follows:
start with the transitive tournament $TT_{n-2}$, whose vertices are labeled $1,\hdots,n-2$ by decreasing scores.
Then add two special vertices $u_1,u_2$ which are connected by an arbitrarily oriented arc.
For $i\in\{1,2\}$, add arcs from all vertices $2,\hdots,n-2$ to $u_i$.
Finally, connect vertex $1$ to $u_i$ by an arc $(u_i,1)$, for $i=1,2$.
The resulting tournament $U_n$, depicted in Fig.~\ref{fig:fewminimalfvs}, has exactly three minimal \FVSs, namely $\{u_1,u_2\},\{1\}$ and $\{2,\hdots,n-2\}$.
\end{proof}


%
%
%

\section{Lower Bound on the Maximum Number of Minimal FVSs}
\label{sec:lowerbound}
We prove a lower bound of $21^{n/7} > 1.5448^n$ on the maximum number of minimal \FVSs of tournaments with $n$ vertices.

Formally, we will bound from below the values of the function $M(n)$ mapping integers $n$ to $\max_{T\in\mathcal T_n}f(T)$.
By convention, set $M(0)=1$.
Note that $M$ is monotonically non-decreasing on its domain:
given any tournament $T\in\mathcal T_n$ and any vertex $v\in V(T)$, for every minimal \FVS $F\in\mathcal F(T[V(T)\setminus\{v\}])$ either $F\in\mathcal F(T)$ or $F\cup\{v\}\in\mathcal F(T)$.
As $T$ and $v$ are arbitrary it follows that $M(n)\geq M(n-1)$.

\tikzset{vertex/.style={minimum size=1mm,circle,fill=black,draw},
         decoration={markings,mark=at position .5 with {\arrow[black,thick]{stealth}}}};

\begin{figure}
 \centering
 \subfloat[Family $U_n$ of strong tournaments with only three minimal \FVSs.]{\label{fig:fewminimalfvs}
 \begin{tikzpicture}
  \node at (-1,-0.7) {$TT_{n-2}$};
  \draw[rounded corners] (-0.75,-0.6) rectangle (4.2,0.6);
  \node (1) at (0,0) [vertex,label=265:$1$] {};
  \node (2) at (1,0) [vertex,label=265:$2$] {};
  \node (3) at (2,0) {$\dots$};
  \node (4) at (3,0) [vertex,label=275:$n-2$] {};
  \node (a) at (1.5,1.5) [vertex,label=above:$u_1$] {};
  \node (b) at (1.5,-1.5) [vertex,label=below:$u_2$] {};
  \draw[postaction={decorate}] (1)--(2);
  \draw[postaction={decorate}] (2)--(3);
  \draw[postaction={decorate}] (3)--(4);
  \draw[postaction={decorate}] (4)--(a);
  \draw[postaction={decorate}] (a)--(1);
  \draw[postaction={decorate}] (3)--(a);
  \draw[postaction={decorate}] (2)--(a);
  \draw[postaction={decorate}] (4)--(b);
  \draw[postaction={decorate}] (b)--(1);
  \draw[postaction={decorate}] (3)--(b);
  \draw[postaction={decorate}] (2)--(b);
  \node at (0,-2.5) {}; 
 \end{tikzpicture}	
 }
 \subfloat[Paley digraph $ST_7$.]{\label{fig:st7}
 \begin{tikzpicture}
  \foreach \x in {1,...,7} {
    \pgfmathparse{\x*360/7};
    \node (a) at (\pgfmathresult:2cm) [vertex,label=\pgfmathresult:\x] {};
    \foreach \y in {1,2,4} {
      \pgfmathparse{(\x+\y)*360/7};
      \draw[postaction={decorate}] (a) -- (\pgfmathresult:2cm);
    }
  }
 \end{tikzpicture}	
 }
 \subfloat[A tournament $pq(T')\in\mathcal T^*_n$ with $f(pq(T'))=2f(T')+1$.]{\label{fig:lowerbeta}
 \begin{tikzpicture}
  \node (p) at (0,0) [vertex,label=left:$p$] {};
  \node (q) at (3,0) [vertex,label=right:$q$] {};
  \node (T) at (1.5,-1.5) {$T'$};
  \draw[postaction={decorate}] (p)--(T);
  \draw[postaction={decorate}] (T)--(q);
  \draw[postaction={decorate}] (q)--(p);
  \node at (0,-3) {}; 
 \end{tikzpicture}	
 }
\caption{Constructions of extremal tournaments.}
\end{figure}
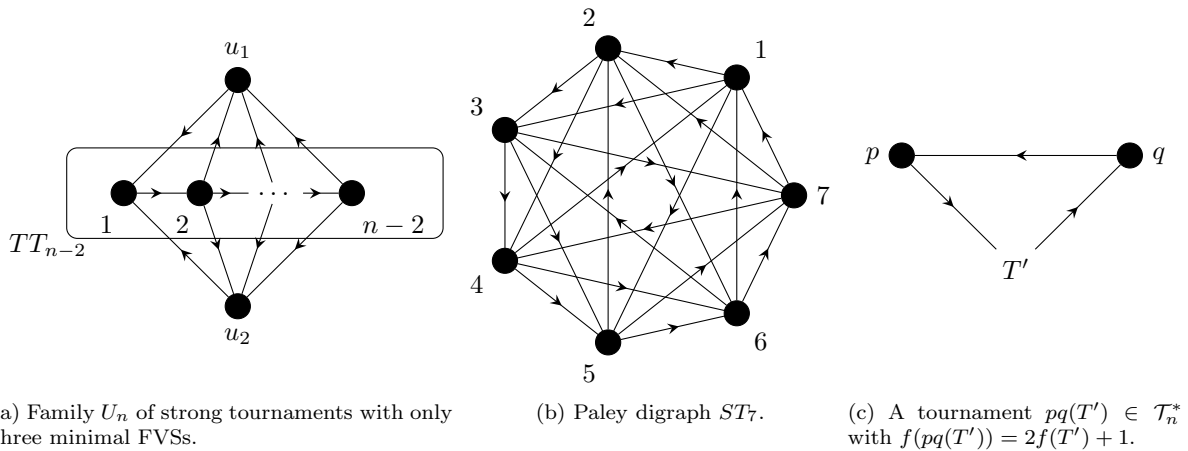

We will now show that there is an infinite family of tournaments on $n=7k$ vertices, for any $k\in\mathbb N$, with $21^{n/7} > 1.5448^n$ minimal \FVSs, improving upon Moon's \cite{Moon1971} bound of $1.4757^n$.
Let $ST_7$ denote the Paley digraph of order 7, i.e. the circular $7$-tournament induced by the set $L = \{1,2,4\}$ of quadratic residues modulo 7 (see Fig. \ref{fig:st7}).
All maximal transitive subtournaments of $ST_7$ are transitive triangles, of which there are exactly 21, as each vertex is the source of 3 distinct transitive triangles.
Thus, all minimal \FVSs for $ST_7$ are minimum \FVSs.
We remark that $ST_7$ is the unique $7$-vertex tournament without any $TT_4$ as subtournament \cite{ReidParker1970}.

\begin{lemma}
\label{thm:manyminimalfvs}
There exists an infinite family of tournaments with $21^{n/7}$ minimal \FVSs.
\end{lemma}
\begin{proof}
Let $k\in\mathbb N$ and form the tournament $T_0=ST_7+\hdots+ST_7$ from $k$ copies of $ST_7\in\mathcal T^*_7$.
Then $T_0\in\mathcal T_n$ for $n=7k$, and the number of minimal \FVSs in $T_0$ is
$f(T_0)=f(ST_7)^k=21^k=21^{n/7}$.
\end{proof}

\section{Upper Bound on the Maximum Number of Minimal FVSs}
\label{sec:upperbound}
We give an upper bound of $\beta^n$, where $\beta =\fvstub$, on the maximum number of minimal \FVSs in any tournament $T\in\mathcal T_n$, for any positive integer $n$.
This improves the bound of $1.7170^n$ by Moon \cite{Moon1971}.
Instead of minimal \FVSs we count maximal transitive subtournaments, and with respect to Observation~\ref{thm:fvsstrongdecomposition} we count the maximal transitive subtournaments of \emph{strong} tournaments.

\medskip

We start with three properties of maximal transitive subtournaments.
First, for a strong tournament $T = (V,A)$ with score sequence $s =(s_1,\hdots,s_n)$ the following holds:
if $TT_k = (V',A')$ is a maximal transitive subtournament of $T$ with $\tau_1(V') = (t)$ then $T[V'\setminus\{t\}]$ is a maximal transitive subtournament of $T[N^+(t)]$.
Hence $f(T)\leq\sum_{v=1}^nM(s_v)
$, where $s_v\leq n-2$ for all $v\in V$.
This allows us to effectively bound $f(T)$ via a recurrence relation.

Second, there cannot be too many vertices with large score.
\begin{lemma}
\label{thm:nummaxscore}
For $n\geq 8$ and $k\in\{0,1,2\}$, any strong tournament $T\in\mathcal T^*_n$ has at most $2(k+1)$ vertices
of score at least $n-2-k$.
\end{lemma}
\begin{proof}
Fix some strong tournament $T\in\mathcal T^*_n$ and $k\in\{0,1,2\}$.
Suppose for contradiction that $T$ contains $2k+3$ vertices with score at least $n-2-k$.
Then the Landau inequalities \eqref{eqn:sbound2} and \eqref{eqn:sbound3} imply the contradiction
\begin{eqnarray*}
2\binom{n}{2} & = & 2\left(\sum_{v=1}^{n-(2k+3)}s_v+\sum_{v=n-(2k+2)}^ns_v\right)\\
              & \geq & 2\left(\binom{n-(2k+3)}{2}+1+(2k+3)(n-2-k)\right) =  n^2-n+2.
\end{eqnarray*}
\end{proof}
For $n\leq 7$, we can explicitly list the strong $n$-vertex tournaments for which the Lemma fails:
the cyclic triangle for $k=0$, the tournaments $RT_5,ST_6$ for $k=1$ and $ST_7$ for $k=2$. $RT_5$ is the regular tournament of order 5 and $ST_6$ is the tournament obtained by arbitrarily removing some vertex from $ST_7$ (defined in the previous section) and all incident arcs.

Third, let $T'$ be a tournament obtained from a tournament $T$ by reversing all arcs of $T$.
Then, $f(T) = f(T')$, whereas the score $s_v(T)$ of each vertex $v$ turns into $s_v(T') = n - 1 - s_v(T)$.
This implies that analyzing score sequences with maximum score $s_n\geq n-1-c$ for some constant $c$ is symmetric to analyzing score sequences with minimum score $s_1\leq c$.

\medskip

We provide a complete proof of the upper bound on the maximum number of minimal feedback vertex set in tournaments.
Our proof that any tournament on $n$ vertices has at most $\beta^n$ maximal transitive subtournaments consists of several parts.
We start by proving the bound for tournaments with few vertices.
The inductive part of the proof first considers tournaments with large maximum score (and symmetrically small minimum score),
and then all other tournaments.

We begin the proof by considering tournaments with up to 10 vertices.
For $n\leq 4$ exact values for $M(n)$ were known before~\cite{Moon1971}.
For $n=5,\hdots,9$ we obtained exact values for $M(n)$ with the help of a computer.
For these values the extremal tournaments obey the following structure:
pick a strong tournament $T'\in\mathcal T^*_{n-2}$ and construct the strong tournament $pq(T') \in \mathcal T^*_n$ by attaching two vertices to $T'$ as in
Fig. \ref{fig:lowerbeta}; namely add vertices $p$ and $q$ to $T'$, and arcs $q \rightarrow p$, and $p \rightarrow t$, $t \rightarrow q$ for each vertex $t$ in $T'$.
Then $f(pq(T'))=2f(T')+1$, as observed by Moon~\cite{Moon1971}.

For $n = 5$, there are exactly two non-isomorphic strong tournaments $QT_5\cong pq(C_3),\linebreak RT_5\in\mathcal T^*_5$. For these, $f(QT_5)=f(RT_5)=M(5)=2\cdot 3+1=7$.
%
%
For $n = 6$, $ST_6$ is the unique tournament from $\mathcal T_6$ with $f(ST_6)=M(6)=12$ minimal \FVSs.
For $n = 7$ the previous section showed $f(ST_7)=21$, and in fact $ST_7$ is the unique $7$-vertex tournament with $M(7)=21$ minimal \FVSs.
For $n \in\{8,9\}$, $ST_n \cong pq(ST_{n-2})$; then $f(ST_n)=M(n)$.
Table~\ref{tab:smallextremal} summarizes that for $n\leq 9$, $M(n)\leq\beta^n$.
\begin{table}
\centering
\caption{Extremal tournaments of up to 9 vertices}
\begin{tabular}{rrrl}
\hline
\noalign{\smallskip}
 $n\;$ & $M(n)\;$      & $\quad M(n)^{1/n} \approx\quad$ & $T\in\mathcal T_n:f(T)=M(n)$\\
\noalign{\smallskip}
\hline
\noalign{\smallskip}
 1\;   & $       1\;$ & $     1.00000\quad$             & $T\in\mathcal T_1$\\
 2\;   & $       1\;$ & $     1.00000\quad$             & $T\in\mathcal T_2$\\
 3\;   & $       3\;$ & $     1.44225\quad$             & $T\in\mathcal T_3\setminus\{TT_3\}$\\
 4\;   & $       3\;$ & $     1.31607\quad$             & $T\in\mathcal T_4\setminus\{TT_4\}$\\
 5\;   & $       7\;$ & $     1.47577\quad$             & $QT_5\cong pq(C_3),RT_5$\\
 6\;   & $      12\;$ & $     1.51309\quad$             & $ST_6\cong ST_7-\{1\}$\\
 7\;   & $      21\;$ & $     1.54486\quad$             & $ST_7$\\
 8\;   & $      25\;$ & $     1.49535\quad$             & $ST_8\cong pq(ST_6)$\\
 9\;   & $      43\;$ & $     1.51879\quad$             & $ST_9\cong pq(ST_7)$\\
\hline
\end{tabular}
\label{tab:smallextremal}
\end{table}

Next, we bound $M(10)$ by means of $M(n)$ for $n\leq 9$.
Let $W$ be a maximal transitive vertex set of $T\in \mathcal T^*_{10}$.
Then either $v^* \in W$ or $v^* \notin W$, where $v^*$ is a vertex with score $s_{10}$.
There are at most $M(s_{10}) \le M(9)$ maximal transitive vertex sets $W$ such that $v^* \in W$ and at most $M(9)$ such sets $W$ for which $v^* \notin W$.
As $(2M(9))^{1/10}=86^{1/10}< 1.5612$, the proof follows for all tournaments with at most 10 vertices.

For the rest of this section we consider tournaments with $n\geq 11$ vertices.
Let $T = (V,A)$ be a strong tournament on $n\geq 11$ vertices and let $s = (s_1,\hdots, s_n)$ be the score sequence of $T$.
We will show that $f(T)\leq \beta^n$.
The proof considers four main cases and several subcases with respect to the minimum and maximum score of the tournament.
To avoid a cumbersome nesting of cases, whenever inside a given case we assume that none of the earlier cases applies.
By $W$ we denote a maximal transitive vertex set of $T$.\\

\noindent
\textbf{Case 1: $\mathbf{s_n = n-2}$.}
Let $b$ be the unique vertex dominating vertex $n$.\\
If $b \notin W$ then $\tau_1(W)=(n)$; there are at most $M(s_n)=M(n-2)$ such~$W$.\\
If $b\in W$ and $n\in W$, then $\tau_1(W \setminus \{b\})=(n)$ as no vertex except $b$ dominates $n$. So, $\tau_2(W)=(b,n)$ and there are at most $M(s_b-1)$ such $W$.
For the last possibility, where $b\in W$ and $n\notin W$,
note that $W$ contains at least one in-neighbor of $b$, otherwise $W$ were not maximal as $n$ could be added.
We consider 4 subcases depending on the score of $b$.
\begin{description}
\item[Case 1.1: $\mathbf{s_b = n-2}$.]
     Let $c$ be the unique vertex dominating $b$.
     As at most 2 vertices have score $n-2$ by Lemma~\ref{thm:nummaxscore}, $s_c \le n-3$.
     We have that $c \in W$, otherwise $W$ would not be maximal as $W\cup \{n\}$ induces a transitive subtournament of $T$.
     As $b$ and its unique in-neighbor $c$ are in $W$, $\tau_2(W)=(c,b)$.
     There are at most $M(s_c-1) \le M(n-4)$ such $W$.
     In total, $f(T)\leq M(n-2) + M(n-3) + M(n-4)\leq \beta^{n-4} + \beta^{n-3} + \beta^{n-2}$ which is at most $\beta^n$ because $\beta \ge 1.4656$.
\end{description}
In the three remaining subcases, all in-neighbors of $b$ have score at most $n-3$: if $c_i\in N^-(b)$ had score $n-2$, then Case 1.1 would apply with $n:=c_i$ and $b:=n$.
\begin{description}
\item[Case 1.2: $\mathbf{s_b = n-3}$.]
     Let $N^-(b):=\{c_1,c_2\}$ such that $c_1 \rightarrow c_2$.
     Then either $\tau_1(W) = (c_1)$ or $\tau_1(W) = (c_2)$; there are at most $2M(n-3)$ such $W$.
     It follows $f(T)\leq M(n-2) + M(n-4) + 2M(n-3)\leq \beta^{n-4} + 2\beta^{n-3} + \beta^{n-2} \leq \beta^n$ as $\beta \ge 1.6181$.
\item[Case 1.3: $\mathbf{s_b = n-4}$.]
     Let $N^-(b):=\{c_1,c_2,c_3\}$.
     Observe that at most $2$ vertices among $N^-(b)$ have score $n-3$, otherwise $T$ is not strong as $N^-(b)\cup \{b,n\}$ induce a strong component.
     Either $\tau_1(W) = (c_1)$ or $\tau_1(W) = (c_2)$ or $\tau_1(W)=(c_3)$; there are at most $2M(n-3)+M(n-4)$ such $W$.
     Thus, $f(T)\leq M(n-2) + M(n-5) + 2M(n-3) + M(n-4)\leq \beta^{n-5} + \beta^{n-4} + 2\beta^{n-3} + \beta^{n-2} \leq \beta^n$ as $\beta \ge 1.6664$.
\item[Case 1.4: $\mathbf{s_b \leq n-5}$.]
     Then there are at most $M(n-1)$ subtournaments not containing $n$.
     It follows $f(T)\leq M(n-2) + M(n-6) + M(n-1)\leq \beta^{n-6} + \beta^{n-2} + \beta^{n-1} \leq \beta^n$ as $\beta \ge 1.6737$.
\end{description}

\noindent
\textbf{Case 2: $\mathbf{s_n = n-3}$.}
Let $b_1,b_2$ be the two vertices dominating $n$ such that $b_1\rightarrow b_2$.
The tree in Fig.~\ref{fig:searchtree} pictures our case distinction.
Its leaves correspond to six different cases, numbered (1)--(6), for membership or non-membership of $n$, $b_1$ and $b_2$ in some maximal transitive vertex set $W$ of $T$.
The cases corresponding to leaves (2) and (4) will be considered later.
Let us now bound the number of possible $W$ for the other cases (1), (3), (5) and (6).\\
\begin{figure}[tbp]
 \centering
 \begin{tikzpicture}
  \tikzstyle{level 1}=[sibling distance=50mm,level distance=15mm]
  \tikzstyle{level 2}=[sibling distance=30mm,level distance=15mm]
  \tikzstyle{level 3}=[sibling distance=25mm,level distance=15mm]

  \node (b1) {$b_1$}
    child {
     node (b2) {$b_2$}
       child {
        node (c1) {(1)$\tau_1(W)=(n) \atop M(n-3)$}
        edge from parent node[left] {$\notin W$}
       }
       child {
        node (n2) {$n$}
          child {
           node (c2) {(2)}
           edge from parent node[left] {$\notin W$}
          }
          child {
           node (c3) {(3)$\tau_2(W)=(b_2,n) \atop M(s_{b_2}-1)$}
           edge from parent node[right] {$\in W$}
          }
        edge from parent node[right] {$\in W$}
       }
     edge from parent node[left] {$\notin W$}
    }
    child {
     node (n) {$n$}
       child {
        node (c4) {(4)}
        edge from parent node[left] {$\notin W$}
       }
       child {
        node (b22) {$b_2$}
          child {
           node (c5) {(5)$\tau_2(W)=(b_1,n) \atop M(s_{b_1}-2)$}
           edge from parent node[left] {$\notin W$}
          }
          child {
           node (c6) {(6)$\tau_3(W)=(b_1,b_2,n) \atop M(s_{b_1}-2)$}
           edge from parent node[right] {$\in W$}
          }
        edge from parent node[right] {$\in W$}
       }
     edge from parent node[right] {$\in W$}
    };
 \end{tikzpicture}
 \caption{\label{fig:searchtree}Different possibilities for a maximal transitive vertex set $W$.}
\end{figure}
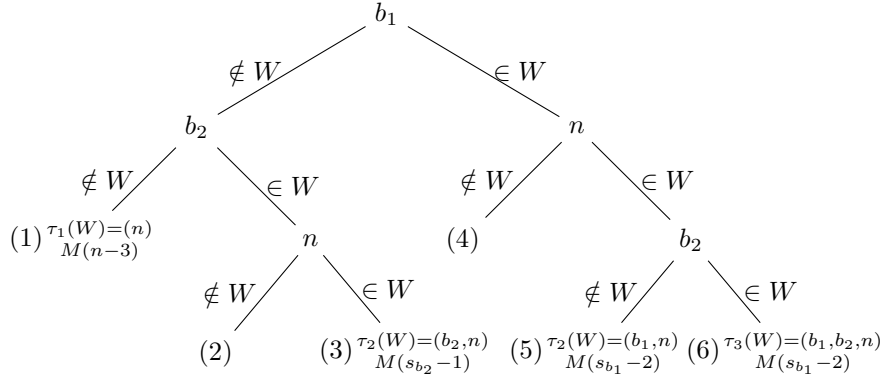

\begin{boldclaim}\label{cl:1356}
 Among all maximal transitive vertex sets $W$ of $T$,
 \begin{itemize}
  \item[(1)] at most $M(n-3)$ are such that $b_1 \notin W$ and $b_2 \notin W$,
  \item[(3)] at most $M(s_{b_2}-1)$ are such that $b_1 \notin W$, $b_2 \in W$ and $n \in W$,
  \item[(5)] at most $M(s_{b_1}-2)$ are such that $b_1 \in W$, $b_2 \notin W$ and $n\in W$, and
  \item[(6)] at most $M(s_{b_1}-2)$ are such that $b_1 \in W$, $b_2 \in W$ and $n\in W$.
 \end{itemize}
\end{boldclaim}
\begin{proof}
If (1) $b_1 \notin W$ and $b_2 \notin W$, then $n\in W$ by maximality of $W$ and $n$ is the source of $T[W]$ as no vertex in $W$ dominates $n$. Thus, there are at most $M(s_n)=M(n-3)$ such $W$.
If (3) $b_1 \notin W$, $b_2 \in W$ and $n \in W$, then $\tau_1(W\setminus \{b_2\})=(n)$. Therefore, $\tau_2(W)=(b_2,n)$ and there are at most $M(s_{b_2}-1)$ such $W$.
If (5) $b_1 \in W$, $b_2 \notin W$ and $n\in W$, then $\tau_2(W)=(b_1,n)$, and as $b_1$ dominates $b_2$, there are at most $M(s_{b_1}-2)$ such $W$.
If (6) $b_1 \in W$, $b_2 \in W$ and $n\in W$, then $\tau_3(W)=(b_1,b_2,n)$, and there are at most $M(s_{b_1}-2)$ such $W$.
\end{proof}

\noindent
To bound the number of subtournaments corresponding to the conditions in leaves (2) and (4), we will consider five subcases depending on the scores of $b_1$ and $b_2$. If $b_1$ and $b_2$ have low scores (Cases 2.4 and 2.5), there are few maximal transitive subtournaments of $T$ corresponding to the conditions in the leaves (3), (5) and (6). Then, it will be sufficient to group the cases (2) and (4) into one case where $n \notin W$ and to note that there are at most $M(n-1)$ such subtournaments. Otherwise, if the scores of $b_1$ and $b_2$ are high (Cases 2.1 -- 2.3), we use that in (2), some vertex of $N^-(b_2)$ is the source of $W$. If this were not the case, $W$ would not be maximal as $W \cup \{n\}$ would induce a transitive tournament. Similarly, in (4) some vertex of $N^-(b_1)$ is the source of $W$ if $b_2 \notin W$.

Let $c_1,\hdots,c_{|N^-(b_1)|}$ be the in-neighbors of $b_1$ such that $c_i \rightarrow c_{i+1}$ for all $i \in \{1,\hdots,\linebreak |N^-(b_1)|-1\}$ (such an ordering exists as every tournament has a Hamiltonian path by \cite{Redei34}; see \cite{BangJensenGutin2002} for a short proof) and let $d_1,\hdots,d_{|N^-(b_2)|-1}$ be the in-neighbors of $b_2$ besides $b_1$ such that $d_i \rightarrow d_{i+1}$ for all $i \in \{1,\hdots,|N^-(b_2)|-2\}$.

Let us first bound the number of subtournaments satisfying the conditions of (2) depending on $s_{b_2}$.

\begin{boldclaim}\label{cl:2.n-3}
 If $s_{b_2}=n-3$, there are at most $M(s_{d_1}-1)$ maximal transitive vertex sets $W$ such that $b_1 \notin W$, $b_2 \in W$ and $n \notin W$.
\end{boldclaim}
\begin{proof}
As mentioned above, some in-neighbor of $b_2$ is the source of $W$. As $s_{b_2}=n-3$, $N^-(b_2)\setminus \{b_1\} = \{d_1\}$. Thus, $\tau_2(W)=(d_1,b_2)$ and there are at most $M(s_{d_1}-1)$ such tournaments.
\end{proof}

\begin{boldclaim}\label{cl:2.n-4}
If $s_{b_2}=n-4$, there are at most~$M(n-5)+2M(s_{d_1}-2)$ maximal transitive vertex sets $W$ such that $b_1 \notin W$, $b_2 \in W$ and $n \notin W$.
\end{boldclaim}
\begin{proof}
If $d_1 \notin W$ then $\tau_2(W)=(d_2,b_2)$ and there are at most~$M(s_{b_2}-1) = M(n-5)$ such $W$. Otherwise, $d_1 \in W$ and either $d_2 \notin W$ in which case $\tau_2(W)=(d_1,b_2)$, or $d_2 \in W$ in which case $\tau_3(W)=(d_1,d_2,b_2)$. There are at most $2M(s_{d_1}-2)$ such $W$.
\end{proof}

The next step is to bound the number of subtournaments satisfying the conditions of (4) depending on $s_{b_1}$.

\begin{boldclaim}\label{cl:4.n-3}
If $s_{b_1}=n-3$, the number of maximal transitive vertex sets $W$ such that $b_1 \in W$ and $n \notin W$ is at most $2M(n-5)+M(n-4)$ if $b_2$ dominates no vertex of $N^-(b_1)$, and otherwise at most $2M(n-5)+M(n-4)+M(n-6)$ if $s_{b_2}=n-3$ and at most $2M(n-5)+M(n-4)+3M(n-7)$ if $s_{b_2}=n-4$.
\end{boldclaim}
\begin{proof}
If $N^-(b_1)\cap W \not = \emptyset$, then $c_1$ or $c_2$ is the source of $W$. The number of subsets $W$ such that $c_1 \notin W$, and thus $\tau_2(W)=(c_2,b_1)$, is at most $M(s_{c_2}-1) \le M(n-4)$. The number of subsets $W$ such that $c_1 \in W$, and thus $\tau_3(W)=(c_1,c_2,b_1)$ or $\tau_2(W)=(c_1,b_1)$, is at most $2M(s_{c_1}-2)\le 2M(n-5)$. If, on the other hand, $N^-(b_1) \cap W = \emptyset$, then $\tau_1(W)=(b_1)$ and some in-neighbor of $b_2$ is the source of $T[W\setminus \{b_1\}]$, otherwise $W$ is not maximal as $n$ can be added. Also note that $b_2$ dominates some vertex of $N^-(b_1)$ (we have $N^-(b_2)\setminus N^-(b_1) \not = \emptyset$ as $N^-(b_1)\cap W = \emptyset$ but $N^-(b_2)\cap W \not = \emptyset$). If $s_{b_2}=n-3$, we upper bound the number of such subsets $W$ by $M(s_{b_1}-3)=M(n-6)$ as $\tau_3(W)=(b_1,d_1,b_2)$. If $s_{b_2}=n-4$, we have that $\tau_4(W)=(b_1,d_1,d_2,b_2)$, $\tau_3(W)=(b_1,d_2,b_2)$ or $\tau_3(W)=(b_1,d_1,b_2)$. Thus, there are at most $3M(s_{b_1}-4)=3M(n-7)$ possible $W$ such that $N^-(b_1) \cap W = \emptyset$ if $s_{b_1}=n-3$ and $s_{b_2}=n-4$. Summarizing, there are at most $2M(n-5)+M(n-4)$ subsets $W$ if $b_2$ dominates no vertex of $N^-(b_1)$, and otherwise at most $2M(n-5)+M(n-4)+M(n-6)$ subsets $W$ if $s_{b_2}=n-3$ and at most $2M(n-5)+M(n-4)+3M(n-7)$ subsets $W$ if $s_{b_2}=n-4$.
\end{proof}

\begin{boldclaim}\label{cl:4.n-4}
If $s_{b_1}=n-4$ and $s_{b_2}=n-3$, the number of maximal transitive vertex sets $W$ such that $b_1 \in W$ and $n \notin W$ is
\begin{itemize}
 \item
  at most $M(n-7)+\sum_{c\in N^-(b_1)}2M(s_c-2)$ if $T[N^-(b_1)]$ is a directed cycle,
 \item
  at most $\max\{M(n-3)+M(n-4)+M(n-5) ; M(n-5)+6M(n-6)\}$ if $T[N^-(b_1)]$ is transitive and $d_1 \in N^-(b_1)$, and
 \item
  at most $M(n-3)+M(n-4)+M(n-5)+M(n-7)$ if $T[N^-(b_1)]$ is transitive and $d_1 \notin N^-(b_1)$.
\end{itemize}
\end{boldclaim}
\begin{proof}
If $c_3 \rightarrow c_1$, then $W$ intersects $N^-(b_1)$ in at most $2^3-1=7$ possible ways ($N^-(b_1) \subseteq W$ would induce a cycle in $T[W]$). In one of them, $N^-(b_1) \cap W = \emptyset$, which implies $\tau_3(W)=(b_1,d_1,b_2)$; there are at most $M(s_{b_1}-3)=M(n-7)$ such $W$. For each $c \in N^-(b_1)$, there are 2 possibilities where $\tau_1(W)=(c)$; one where $\tau_2(W)=(c,b_1)$ and one where $\tau_3(W)=(c,y,b_1)$ where $y$ is the out-neighbor of $c$ in $N^-(b_1)$; there are $2M(s_c-2)$ such $W$ for each choice of $c$. In total, there are at most $M(n-7)+\sum_{c\in N^-(b_1)}2M(s_c-2)$ possible $W$. 

If, on the other hand, $c_1 \rightarrow c_3$, first assume that $s_{c_1} \le n-3$, $s_{c_2} \le n-4$, and $s_{c_3} \le n-5$. Then either some vertex of $N^-(b_1)$ is the source of $W$ (at most $M(n-3)+M(n-4)+M(n-5)$ possibilities for $W$), or $\tau_3(W)=(b_1,d_1,b_2)$ (at most $M(n-7)$ possibilities for $W$). Otherwise, it must be that $s_{c_1}\le n-3$, $s_{c_2}\le n-4$, $s_{c_3} = n-4$ and that $d_1=c_3$. Then, $\tau_2(W)=(c_3,b_1)$, $\tau_2(W)=(c_2,b_1)$, $\tau_3(W)=(c_2,c_3,b_1)$, $\tau_2(W)=(c_1,b_1)$, $\tau_3(W)=(c_1,c_2,b_1)$, $\tau_3(W)=(c_1,c_3,b_1)$, or $\tau_4(W)=(c_1,c_2,c_3,b_1)$; there are at most $M(n-5)+6M(n-6)$ such $W$. In total, if $d_1 \in N^-(b_1)$, the number of possible $W$ can be upper bounded by $\max\{M(n-3)+M(n-4)+M(n-5) ; M(n-5)+6M(n-6)\}$, and if $d_1 \notin N^-(b_1)$, the number of possible $W$ can be upper bounded by $M(n-3)+M(n-4)+M(n-5)+M(n-7)$.
\end{proof}

\noindent
Armed with Claims~\ref{cl:2.n-3}--\ref{cl:4.n-4}, we now analyze the five subcases of Case 2, depending on the scores of $b_1$ and $b_2$.\\

\noindent
\textbf{Case 2.1: $\mathbf{s_{b_1} = n-3,s_{b_2} = n-3}$.}
By Claim~\ref{cl:2.n-3}, the number of maximal transitive vertex sets $W$ such that $b_1,n \notin W$ and $b_2\in W$ (leaf (2) in Fig.~\ref{fig:searchtree}) is at most $M(n-4)$. By Claim~\ref{cl:4.n-3}, the number of maximal transitive vertex sets $W$ such that $b_1,n \notin W$ and $b_2\in W$ (leaf (4) in Fig.~\ref{fig:searchtree}) is at most $2M(n-5)+M(n-4)$, at most $2M(n-5)+M(n-4)+M(n-6)$, or at most $2M(n-5)+M(n-4)+3M(n-7)$.
Combined with Claim~\ref{cl:1356},

\begin{align*}
f(T) &\leq \max
\begin{cases}
M(n-3)+M(n-4)+M(n-4)+(2M(n-5)\\
  \quad +M(n-4))+M(n-5) +M(n-5)\\
  \hfill \leq 4\beta^{n-5}+3\beta^{n-4}+\beta^{n-3}\leq\beta^n \text{ as } \beta \ge 1.6314 \enspace ,\\
M(n-3)+M(n-4)+M(n-4)+(2M(n-5)\\
  \quad +M(n-4)+M(n-6)) +M(n-5)+M(n-5)\\
  \hfill \leq \beta^{n-6}+4\beta^{n-5}+3\beta^{n-4}+\beta^{n-3}\leq\beta^n \text{ as } \beta \ge 1.6516\enspace ,\\
M(n-3)+M(n-4)+M(n-4)+(2M(n-5)\\
  \quad +M(n-4)+3M(n-7)) +M(n-5)+M(n-5)\\
  \hfill \leq 3\beta^{n-7}+4\beta^{n-5}+3\beta^{n-4}+\beta^{n-3}\leq\beta^n \text{ as } \beta \ge 1.6666 \enspace .
\end{cases}
\end{align*}

\noindent
\textbf{Case 2.2: $\mathbf{s_{b_1} = n-3,s_{b_2} = n-4}$.}
If $c_1 \rightarrow b_2$ and $c_2 \rightarrow b_2$, then $b_1 \notin W$ and $b_2 \in W$ implies that some in-neighbor $c$ of $b_1$ is in $W$, otherwise $W\cup\{b_1\}$ would induce a transitive tournament. But then, $n \notin W$, otherwise $\{c,b_2,n\}$ induces a directed cycle. This means that no maximal transitive vertex set $W$ satisfies the conditions of leaf (3) in Fig.~\ref{fig:searchtree}. We bound the possible $W$ corresponding to leaves (2)+(4) by $M(n-1)$ and obtain
\begin{align*}
 f(T) &\le M(n-3)+M(n-1)+M(n-5)+M(n-5)\\
      & \le 2\beta^{n-5}+\beta^{n-3}+\beta^{n-1}\leq\beta^n \text{ as } \beta \ge 1.6440 \enspace .
\end{align*}

Otherwise, there is some vertex $c \in N^-(b_1)$ such that $b_2 \rightarrow c$. Then, the number of $W$ in leaf (6) of Fig.~\ref{fig:searchtree} is upper bounded by $M(s_{b_2}-2)=M(n-6)$, and by Claims~\ref{cl:2.n-4} and \ref{cl:4.n-3} those in leaves (2) and (4) are upper bounded by $M(n-5)+2M(s_{d_1}-2)$ and $2M(n-5)+M(n-4)+3M(n-7)$, respectively. Thus,
\begin{align*}
 f(T) &\le M(n-3)+(M(n-5)+2M(n-5))+M(n-5)+(2M(n-5)\\
      & \quad +M(n-4)+3M(n-7))+M(n-5)+M(n-6)\\
      & \le 3\beta^{n-7}+\beta^{n-6}+7\beta^{n-5}+\beta^{n-4}+\beta^{n-3}\leq\beta^n \text{ as } \beta \ge 1.6740 \enspace .
\end{align*}

\noindent
\textbf{Case 2.3: $\mathbf{s_{b_1} = n-4, s_{b_2} = n-3}$.}
By Claim~\ref{cl:2.n-3}, at most $M(n-4)$ subsets $W$ correspond to leaf (2) in Fig.~\ref{fig:searchtree}.
If $N^-(b_1)$ induces a directed cycle, Claim~\ref{cl:4.n-4} upper bounds the number of subsets corresponding to leaf (4) by $M(n-7)+2M(n-6)+4M(n-5)$ as at most 2 vertices except $b_2$ and $n$ have score $n-3$ by Lemma~\ref{thm:nummaxscore}. Together with Claim~\ref{cl:1356}, this gives
\begin{align*}
 f(T) &\le M(n-3)+M(n-4)+M(n-4)+(M(n-7)+2M(n-6)\\
      & \quad +4M(n-5))+M(n-6)+M(n-6)\\
      & \le \beta^{n-7}+4\beta^{n-6}+4\beta^{n-5}+2\beta^{n-4}+\beta^{n-3}\leq\beta^n \text{ as } \beta \ge 1.6670 \enspace .
\end{align*}
Otherwise, $c_1 \rightarrow c_3$. If $d_1 \rightarrow b_1$, then Claim~\ref{cl:4.n-4} upper bounds the number of subsets corresponding to leaf (4) by $M(n-3)+M(n-4)+M(n-5)$ or $M(n-5)+6M(n-6)$.
Then,
\begin{align*}
f(T) &\leq \max
\begin{cases}
M(n-3)+M(n-4)+M(n-4)+(M(n-3)\\
  \quad +M(n-4)+M(n-5))+M(n-6)+M(n-6) \\
  \hfill \leq 2\beta^{n-6}+\beta^{n-5}+3\beta^{n-4}+2\beta^{n-3}\leq\beta^n \text{ as } \beta \ge 1.6632,\\
M(n-3)+M(n-4)+M(n-4)+(M(n-5)\\
  \quad +6M(n-6))+M(n-6) + M(n-6)\\
  \hfill \leq 8\beta^{n-6}+\beta^{n-5}+2\beta^{n-4}+\beta^{n-3}\leq\beta^n \text{ as } \beta \ge 1.6396 \enspace .\\
\end{cases}
\end{align*}
Otherwise, $b_1 \rightarrow d_1$.
For the possible $W$ with $b_1,b_2,n \in W$, none of $N^-(b_1)\cup \{d_1\}$ is in $W$ as these vertices all create cycles with $b_1,b_2,n$.
Thus, the number of possible subsets $W$ corresponding to leaf (6) is upper bounded by $M(s_{b_1}-3)=M(n-7)$. Then, by Claims~\ref{cl:1356} and~\ref{cl:4.n-4},
\begin{align*}
 f(T) &\le M(n-3)+M(n-4)+M(n-4)+(M(n-3)+M(n-4)\\
      & \quad +M(n-5)+M(n-7))+M(n-6)+M(n-7)\\
      & \le 2\beta^{n-7}+\beta^{n-6}+\beta^{n-5}+3\beta^{n-4}+2\beta^{n-3}\leq\beta^n \text{ as } \beta \ge 1.6672 \enspace .
\end{align*}

\noindent
\textbf{Case 2.4: $\mathbf{s_{b_1} = n-4,s_{b_2} \leq n-4}$.}
By grouping leaves (2) and (4) into one possibility where $n \notin W$, Claim~\ref{cl:1356} upper bounds the number of such maximal transitive vertex sets by
\begin{align*}
 f(T) &\le M(n-3)+M(n-1)+M(n-5)+M(n-6)+M(n-6)\\
      & \le 2\beta^{n-6}+\beta^{n-5}+\beta^{n-3}+\beta^{n-1}\leq\beta^n \text{ as } \beta \ge 1.6570 \enspace .
\end{align*}

\noindent
\textbf{Case 2.5: $\mathbf{s_{b_1} \le n-5}$.}
By grouping leaves (2) and (4) into one possibility where $n \notin W$, Claim~\ref{cl:1356} upper bounds the number of such maximal transitive vertex sets by
\begin{align*}
 f(T) &\le M(n-3)+M(n-1)+M(n-4)+M(n-7)+M(n-7)\\
      & \le 2\beta^{n-7}+\beta^{n-4}+\beta^{n-3}+\beta^{n-1}\leq\beta^n \text{ as } \beta \ge 1.6679 \enspace .
\end{align*}

\noindent
\textbf{Case 3: $\mathbf{s_n \leq n-4}$.}
We may assume that the score sequence $s=s(T)$ satisfies
\begin{equation}
\label{eqn:sbound4}
3\leq s_1\leq \hdots\leq s_n\leq n-4.
\end{equation}
Let $S_n$ be the set of all score sequences that are feasible for \eqref{eqn:sbound2}--\eqref{eqn:sbound4}.
The set $S_n$ serves as domain of the linear map $G:S_n\rightarrow\mathbb R_+,s\mapsto\sum_{v=1}^ng(s_v)$
with the strictly convex terms $g:c\mapsto\beta^c$.
Furthermore, for all $n\ge 11$, we define a special score sequence $\sigma(n)$, whose membership in $S_n$ is easy to verify:
\begin{align*}
\sigma(n):=
\begin{cases}
(3,3,3,3,3,5,7,7,7,7,7) & \text{if } n = 11\enspace ,\\
(3,3,3,3,3,3,8,8,8,8,8,8) & \text{if } n=12\enspace ,\\
(3,3,3,3,3,3,6,9,9,9,9,9,9) & \text{if } n = 13\enspace , \text{ and}\\
(3,3,3,3,3,3,4,7,8,\hdots,n-9,n-8,n-5,\\ \quad \quad ~n-4,n-4,n-4,n-4,n-4,n-4) & \text{if } n \ge 14\enspace .
\end{cases}
\end{align*}

\begin{lemma}
\label{thm:reclemma}
For $n\geq 11$, the sequence $\sigma(n)$ maximizes the value of $G$ over all sequences in $S_n$:
$G(s)\leq G(\sigma(n))$ for all $s\in S_n.$
\end{lemma}
Once Lemma~\ref{thm:reclemma} is proved we can bound $f(T)$, for $s=s(T)\in S_n$, from above via
\begin{align}
\label{eqn:uppernlarge}
f(T)\leq G(s)\leq G(\sigma(n))
=\begin{cases}
5\beta^3+\beta^5+5\beta^7,&\text{if } n=11\enspace ,\\
6\beta^3+6\beta^8,&\text{if } n=12\enspace ,\\
6\beta^3+\beta^6+6\beta^9,&\text{if } n=13\enspace ,\\
6\beta^3+\beta^4+\frac{\beta^{n-7}-\beta^7}{\beta-1}+\beta^{n-5}+6\beta^{n-4}\\ \quad\quad
 \leq \frac{\beta^{n-7}}{\beta-1}+\beta^{n-5}+6\beta^{n-4}, &\text{if } n\geq 14 \enspace,
\end{cases}
\end{align}
which is at most $\beta^n$ as $\beta \ge 1.6259$.
To prove Lemma~\ref{thm:reclemma}, we choose any sequence $s\in\mbox{argmax}_{s'\in S_n} G(s')$ and then show that $s=\sigma(n)$.
Recall that $s_1\geq 3$ and $s_n\leq n-4$, and set $s_1^*=3,s_n^*=n-4$.

\begin{boldclaim}
\label{thm:claim1}
If some score $c$ appears more than once in $s$, then $c\in\{s_1^*,s_n^*\}$.
\end{boldclaim}
\begin{proof}
For contradiction, suppose that $s_1^* < s_u=s_v=c<s_n^*$ for two vertices $u$ and $v$ such that $1\leq u<v\leq n$.
First, suppose there exists an integer $k\in\{u,\hdots,v-1\}$ satisfying \eqref{eqn:sbound2} with equality:
\begin{equation}
\label{eqn:sboundequality}
\sum_{v=1}^ks_v=\binom{k}{2}+1 \enspace .
\end{equation}
Then \eqref{eqn:sbound2}, \eqref{eqn:sbound3} and Lemma~\ref{thm:nummaxscore} imply $8\leq k\leq n-9$, so $k\notin \{s_1^*,s_n^*\}$.
The choice of $k$ among vertices of equal score $c$ now yields
\begin{equation}
\label{eqn:sboundequalscore}
s_{k+1}=s_k=\sum_{v=1}^ks_v-\sum_{v=1}^{k-1}s_v\leq\binom{k}{2}+1-\binom{k-1}{2}-1=k-1 \enspace .
\end{equation}
This however contradicts \eqref{eqn:sbound2}:
\begin{equation*}
\sum_{v=1}^{k+1}s_v\leq\binom{k}{2}+1+(k-1)=\binom{k+1}{2} \enspace .
\end{equation*}
It is thus asserted that no vertex $k$ with property \eqref{eqn:sboundequality} exists.
The score sequence $s'$ differing from $s$ only in $s_u'=s_u-1=c-1,s_v'=s_v+1=c+1$, therefore belongs to $S_n$.
So apply the function $G$ to it, and use the strict convexity of $g$:
\begin{equation*}
G(s')-G(s)=(g(c+1)-g(c))-(g(c)-g(c-1))>0 \enspace .
\end{equation*}
This contradicts the choice of $s$ as a maximizer of $G$, and establishes Claim \ref{thm:claim1}.
\end{proof}

\begin{boldclaim}\label{thm:claim2}
The values $s_1^*=3$ and $s_n^*=n-4$ each appear between two and six times as scores in the sequence $s$.
\end{boldclaim}
\begin{proof}
By Lemma~\ref{thm:nummaxscore}, $s_n^*$ is the score of no more than $6$ vertices. By symmetry, $s_1^*$ is the score of no more than $6$ vertices.
As a consequence of Claim \ref{thm:claim1}, together $s_1^*$ and $s_n^*$ appear at least eight times in $s$.
Hence there are at least two vertices of score $s_1^*$ and at least two vertices of score $s_n^*$.
\end{proof}

\begin{boldclaim}\label{thm:claim3}
If $n \ge 12$, each of $s_1^*$ and $s_n^*$ is the score of exactly six of the vertices.
\end{boldclaim}
\begin{proof}
Assuming this were not the case for $s_1^*$, by Claim \ref{thm:claim2} it would be the score of two to five vertices.
Hence there exists a vertex $a\in\{3,\hdots,6\}$ with score $s_a>s_1^*$.
It holds $s_n^*=n-4>s_a+1$, which is obvious if $n \ge 13$ and follows from \eqref{eqn:sbound3} if $n=12$.
So there must be two scores in $s$ larger than $s_a$, precisely $s_a<s_{a+1}<s_{a+2}$.
Observe that the sequence $s'=(s_1,\hdots,s_{a-1},s_a-1,s_{a+1}+1,s_{a+2},\hdots,s_n)$ is a member of $S_n$.
The same argument on strict convexity of $g$ as in Claim \ref{thm:claim1} gives
\begin{equation*}
G(s')-G(s) = (g(y+1)-g(y))-(g(x)-(g(x-1))>0
\end{equation*}
for $x=s_a<s_{a+1}=y$, again contradicting the choice of $s$ as a maximizer of $G$.
Consequently, the sequence $s$ starts with six scores $s_1^*$.
By symmetry, the same argumentation also applies for $s_n^*$, proving the claim.
\end{proof}

\begin{boldclaim}\label{thm:claim4}
If $n = 11$, each of $s_1^*$ and $s_n^*$ is the score of exactly five of the vertices.
\end{boldclaim}
\begin{proof}
As all scores are between $3$ and $7$, at most $5$ vertices have score $3$ and at most $5$ vertices have score $7$ by \eqref{eqn:sbound3}. Assume less than $5$ vertices have score $s_1^*$. By Claim \ref{thm:claim2}, $s_1^*$ is the score of two to four vertices.
Hence there exists a vertex $a\in\{3,4,5\}$ with score $s_a>s_1^*$. Thus,
$s_n^*=7>a+1$.
So there must be two scores in $s$ larger than $s_a$, precisely $s_a<s_{a+1}<s_{a+2}$.
To conclude we construct a sequence $s'$ with $G(s')>G(s)$ exactly as in the proof of Claim~\ref{thm:claim3}.
\end{proof}

\begin{boldclaim}
It holds $s=\sigma(n)$.
\label{thm:claim5}
\end{boldclaim}
\begin{proof}
If $n=11$, $s$ has $5$ vertices of score $3$ and $5$ vertices of score $7$ by Claim~\ref{thm:claim4}. As, $\sigma(11)$ is the only such sequence not contradicting \eqref{eqn:sbound3}, the claim holds for $n=11$. Similarly, $\sigma(n)$ is the only sequence not contradicting \eqref{eqn:sbound2} and Claim~\ref{thm:claim3} if $12 \le n \le 13$. Suppose now that $n \ge 14$.
There are $n-12$ elements of $s$ being different from both $s_1^*$ and $s_n^*$, which have a score equal to one of the
$n-8$ numbers in the range $4,\hdots,n-5$.
Symmetry of the map $d\mapsto\binom{n}{d}$ around $d=\frac{n}{2}$ together with \eqref{eqn:sbound3}
means that only pairs $\{h_1,n-1-h_1\}$ with $4\leq h_1<\frac{n-1}{2}$ and
$\{h_2,n-1-h_2\}$ with $5\leq h_2<\frac{n-1}{2}$ of scores are missing in $s$.
Moreover, \eqref{eqn:sbound2} requires $h_1,h_2 < 7$, for otherwise $k=8$ violates this relation.
Since $s$ was chosen to be a maximizer of $G$, this leaves $h_1=5$ and $h_2=6$.
Thus $s=\sigma(n)$, completing the proof of the claim and of Lemma \ref{thm:reclemma}.
\end{proof}
%

All cases taken together imply the following upper bound on the number of maximal transitive subtournaments.
\begin{theorem}
\label{thm:combupperbound}
Any strong tournament $T\in\mathcal T^*_n$ has at most $\fvstub^n$ maximal transitive subtournaments.
\end{theorem}

Moon \cite{Moon1971} already observed that the following limit exists.

\begin{corol}
 It holds $1.5448 \le \lim_{n \rightarrow \infty} (M(n))^{1/n} \le \fvstub$.
\end{corol}

We conjecture that the Paley digraph of order 7, $ST_7$, plays the same role for \FVSs in tournaments as triangles play for independent sets in graphs, i.e. that the tournaments $T$ maximizing $(f(T))^{1/|V(T)|}$ are exactly those whose factors are copies of $ST_7$.

\section{Polynomial-Delay Enumeration in Polynomial Space}
\label{sec:polydelaypolyspace}
In this section, we give a polynomial-space algorithm for the enumeration of the minimal \FVSs in a tournament with polynomial delay.

Let $T = (V,A)$ be a tournament with $V = \{v_1,\hdots,v_n\}$, and for each $i = 1,\hdots,n$ let $T_i = T[\{v_1,\hdots,v_i\}]$.
For a vertex set $X$, we write $\chi_X(i)=1$ if $v_i\in X$ and $\chi_X(i)=0$ otherwise.
Let $<$ denote the total order on $V$ induced by the labels of the vertices.
For vertex sets $X,Y\subseteq V$, say that $X$ is \emph{lexicographically smaller} than $Y$ and write $X\prec Y$ if for the minimum index $i$ for which $\chi_X(i)\not=\chi_Y(i)$ it holds that $v_i\in X$.
Because $X$ and $Y$ are totally ordered by the restriction of $<$ to $X$ and~$Y$, respectively, $\prec$ is also a total order and each collection of subsets of $V$ has a unique \emph{lexicographically smallest} element.

The algorithm enumerates the maximal acyclic vertex sets of $T$.
It performs a depth-first search in a tree $\mathcal T$ with the maximal acyclic vertex sets of $T$ as leaves, whose forward and backward edges are constructed ``on the fly''.
The depth of $\mathcal T$ is $|V|$, and we refer to the vertices of $\mathcal T$ as \emph{nodes}. The algorithm only needs to keep in memory the path from the root to the current node in the tree and all the children of the nodes on this path.
Each node at level $j$ is labeled by a maximal acyclic vertex set $J$ of $T_j$.
As for its children, there are two cases.
In case $J\cup\{v_{j+1}\}$ is acyclic then $J$'s only child is $J\cup\{v_{j+1}\}$.
In case $J\cup\{v_{j+1}\}$ is not acyclic then $J$ has at least one and at most $\lfloor j/2 \rfloor + 1$ children.
Let $L_J = (v^1,v^2,\hdots,v^{|J|})$ be a labeling of the vertices in $J$ such that $(v^r,v^s)\in A$ for all $1\leq r < s\leq j$; we view $L_J$ as a sequence of vertices.
The children of $J$ are as follows.
The first child $J^0$ is a copy of $J$, and is always present.
The potential other children are, for $1\le z \le |J|+1$,
\begin{equation*}
  J^z = \{v^i \in J \mid i<z \wedge v^i \rightarrow v_{j+1}\} \cup \{v_{j+1}\} \cup \{v^i \in J \mid i\ge z \wedge v_{j+1} \rightarrow v^i\}
\end{equation*}
where set $J^z$ is a potential child of $J$ only if $J^z$ is a maximal acyclic vertex set in $T_{j+1}$.
To check the maximality of $J^z$, compute a transitive order of its vertices.
The set $J^z$ is not maximal if there exists a vertex $u\in \{v_1, \dots, v_j \} \setminus J^z$ such that there exists $k, 0\le k\le |J^z|$,
such that the first $k$ vertices in the transitive order dominate $u$ and $u$ dominates all other vertices from $J^z$.
The maximality check can therefore be done in $O(n^2)$ time.

Note how we try to insert $v_{j+1}$ at every possible position in $J$ to compute $J^z$. However, only at most $\lfloor j/2 \rfloor+1$ positions make sense for $v_{j+1}$: before $v^1$ if $v_{j+1} \rightarrow v^1$, between $v^i$ and $v^{i+1}$ if $v^i \rightarrow v_{j+1} \rightarrow v^{i+1}$, where $1\le i\le |J|-1$, and after $v^{|J|}$ if $v^{|J|}\rightarrow v_{j+1}$; all other positions do not give maximal acyclic vertex sets and should not be generated in an actual implementation.
Note that $J^z$ may be a potential child of several sets on the same level in $\mathcal T$.
Of all these sets, $J^z$ is made the child only of the lexicographically smallest such set.
To determine whether $J$ is the lexicographically smallest such set, we compute by a greedy algorithm the lexicographically smallest maximal acyclic vertex set $H = H(J^z)$ of $T_j$ which contains $J^z\setminus\{v_{j+1}\}$ as a subset.
That is, we iteratively build the set $H$ by setting
\begin{align*}
  H_0 & = J^z\setminus\{v_{j+1}\},\displaybreak[0]\\
  H_i & = \begin{cases}
    H_{i-1}\cup\{v_i\},&\mbox{if}~H_{i-1}\cup\{v_i\}~\mbox{is acyclic},\\
    H_{i-1},&\mbox{otherwise},
  \end{cases}
  \qquad i = 1,\hdots,j,\displaybreak[0]\\
  H & = H_j \enspace.
\end{align*}
Then we make $J^z$ a child of the node labeled $J$ only if $H = J$. The set $H$ can be computed in $O(n^2)$ time.
This completes the description of the algorithm.

\medskip

To show that the algorithm is correct, we prove that for every maximal acyclic vertex set $W$ of $T$ there is exactly one leaf in $\mathcal T$ labeled with $W$.
By construction of the algorithm, it suffices to show that at least one leaf is labeled by $W$.
The proof is by induction on the number $n = |V|$ of vertices in $T$.
For $n = 1$ the claim clearly holds, so suppose that $n > 1$ and that the claim is true for all tournaments with fewer vertices.
Then from the induction hypothesis we can conclude that for the induced subtournament $T' := T_{n-1}$ there is a tree $\mathcal T'$ constructed by the above algorithm and a bijection $f'$ from the maximal acyclic vertex sets of $T'$ to the leaves of~$\mathcal T'$.

Let $W$ be a maximal acyclic vertex set of $T$.
If $v_n\notin W$ then $W$ is an acyclic vertex set of $T'$ as removing a vertex from a digraph does not introduce cycles.
In fact, $W$ is a maximal acyclic vertex set of $T'$: for any vertex $v_\ell \in V\setminus (W\cup \{v_n\})$, $T'[W\cup \{v_\ell \}]$ has a cycle as $W$ is a maximal acyclic vertex set for $T$ and $T'[W\cup \{v_\ell \}]=T[W\cup \{v_\ell \}]$.
Hence there exists a leaf $f'(W)$ in $\mathcal T'$ labeled by $W$.
Since $W\cup\{v_n\}$ is not acyclic, by maximality of $W$ for $T$, the algorithm constructs the child $W^0$ of $f'(W)$ labeled by $W$, and that child will be a leaf in the final tree constructed by the algorithm.

If $v_n\in W$, then let $W' = W\setminus\{v_n\}$. So, $W'$ is an acyclic vertex set of~$T'$.
In case $W'$ is maximal for $T'$, there is a leaf $f'(W')$ in $\mathcal T'$ that is labeled by $W'$.
Since $W'\cup\{v_n\}$ is acyclic, the algorithm will create a single child of $f'(W')$ labeled by $W'\cup\{v_n\} = W$, and that child will be a leaf in the final tree constructed by the algorithm.
In case $W'$ is not maximal for $T'$, let $N$ be the lexicographically smallest extension of $W'$ to a maximal acyclic vertex set of $T'$.
Hence there exists a leaf $f'(N)$ in the tree $\mathcal T'$ labeled by $N$.
Observe that the sequence $L_{W'}$ is a subsequence of $L_N$, and that $N\cup\{v_n\}$ is not acyclic.
Hence the algorithm creates children $N^1,N^2,\hdots$, one of which will be labeled by~$W$.

\medskip

We show that the algorithm requires only polynomial space.
We already observed that each node in $\mathcal T$ at level $j$ has at most $\lfloor j/2 \rfloor+1$ children.
For each node we store the maximal acyclic vertex set by which it is labeled.
Because we are traversing $\mathcal T$ in a depth-first-search manner, in each step of the algorithm we only need to save data of $O(n^2)$ nodes: those of the $O(n)$ nodes on the path from the root to the currently active node labeled by $J$, and the $O(n)$ children for each node on this path.

To see that the algorithm runs with polynomial delay, note that the set of all children of a given node in $\mathcal T$ can all be computed in $O(n^3)$ time.
It follows that $\mathcal T$ can be traversed in a depth-first manner with polynomial delay per step of the traversal, and thus the leaves of $\mathcal T$ can be output with only a polynomial delay, bounded by $O(n^4)$.

\begin{theorem}
 The described algorithm enumerates all \FVSs of a tournament with polynomial delay and uses polynomial space.
\end{theorem}

\begin{corol}
\label{thm:minfvspolyspace}
  In a tournament with $n$ vertices a minimum directed feedback vertex set can be found in $O(1.6740^n)$ time and polynomial space.
\end{corol}

\paragraph{Acknowledgment.} We thank Gerhard J. Woeginger for help with the presentation of the results.

{\small
\bibliographystyle{abbrv}

}

\end{document}